\documentclass[journal]{IEEEtran}

\usepackage{color}  
\usepackage{hyperref}
\usepackage{balance}
\usepackage{graphicx}
\usepackage{subfigure}
\usepackage[graphicx]{realboxes}
\usepackage{balance}
\usepackage[ruled, vlined, linesnumbered]{algorithm2e}
\usepackage{array}
\usepackage{multirow}
\usepackage[normalem]{ulem}
\useunder{\uline}{\ul}{}
\usepackage{booktabs}
\usepackage{enumitem}
\usepackage{amsthm} 
\newtheorem{theorem}{Theorem}
\newtheorem{example}{Example}
\newtheorem{definition}{Definition}

\newtheorem{observation}{Observation}
 
\def\BibTeX{{\rm B\kern-.05em{\sc i\kern-.025em b}\kern-.08em
    T\kern-.1667em\lower.7ex\hbox{E}\kern-.125emX}}

\begin{document}
\title{MONA: An Efficient and Scalable Strategy for Targeted $k$-Nodes Collapse}



\author{
  Yuqian~Lv,
  Bo~Zhou,
  Jinhuan~Wang,
  Shanqing Yu,
  and~Qi~Xuan,~\IEEEmembership{Senior~Member,~IEEE}
  \thanks{
    This work was supported in part by the Key R\&D Program of Zhejiang under Grants 2022C01018 and 2021C01117, by the National Natural Science Foundation of China under Grants 62103374, 61973273 and U21B2001, by the National Key R\&D Program of China under Grants 2020YFB1006104, and by the Major Key Project of PCL under Grants PCL2022A03, PCL2021A02, and PCL2021A09.
  }
  \thanks{
    Y. Lv, B. Zhou, J. Wang and S. Yu are with the Institute of Cyberspace Security, College of Information Engineering, Zhejiang University of Technology, Hangzhou 310023, China.
  }
  \thanks{
    B. Zhou is also with the Department of Intelligent Control, Zhejiang Institute of Communications, Hangzhou 311112, China.
  }
  \thanks{
    Q. Xuan is with the Institute of Cyberspace Security, College of Information Engineering, Zhejiang University of Technology, Hangzhou 310023, China, with the PCL Research Center of Networks and Communications, Peng Cheng Laboratory, Shenzhen 518000, China, and also with the Utron Technology Co., Ltd. (as Hangzhou Qianjiang Distinguished Expert), Hangzhou 310056, China (e-mail: xuanqi@zjut.edu.cn).
  }
}
\maketitle

\begin{abstract}

  The concept of $k$-core plays an important role in measuring the cohesiveness and engagement of a network. And recent studies have shown the vulnerability of $k$-core under adversarial attacks. However, there are few researchers concentrating on the vulnerability of individual nodes within $k$-core. Therefore, in this paper, we attempt to study \underline{T}argeted $k$-\underline{N}ode\underline{s} \underline{C}ollapse \underline{P}roblem (TNsCP), which focuses on removing a minimal size set of edges to make multiple target $k$-nodes collapse. For this purpose, we first propose a novel algorithm named MOD for candidate reduction. Then we introduce an efficient strategy named MONA, based on MOD, to address TNsCP. Extensive experiments validate the effectiveness and scalability of MONA compared to several baselines. An open-source implementation is available at \emph{\url{https://github.com/Yocenly/MONA}}.

\end{abstract}


\section{Introduction}

\IEEEPARstart{N}{etworks} or graphs have become integral parts of our daily lives, playing significant roles in various complex systems, e.g., social networks, biological networks, and transport networks. Due to simplicity and efficiency, the concept of $k$-core has gained prominence as a crucial metric for capturing the structural engagement of networks. The $k$-core is a maximal induced subgraph where each node has its degree of at least $k$. The presence of nodes in $k$-core depends on their neighboring relationships in the subgraph. Thus, removing a subset of nodes or edges from the $k$-core may result in the detachment of other nodes within it. For example, Zhou et al. \cite{zhou2022attacking} discovered that the removal of a small number of edges may severely devastate the structure of the $k$-core. And Chen et al. \cite{chen2021edge} focused on the $k$-core minimization problem and proposed effective algorithms to cover it. Medya et al. \cite{ijcai2020p480} employed the \textit{Shapley} value, a cooperative game-theoretic concept, to address the problem of $k$-core minimization.



\textbf{Motivations.} Despite the effectiveness of above mentioned methods, they all attack the $k$-core from a global perspective. They pay little attention to the perturbations that affect individual nodes in the $k$-core. According to the research of Zhang et al. \cite{zhang2017finding}, the removal of critical nodes in $k$-core may significantly break down network engagement. However, the perturbations of removing nodes are infeasible in practice and may result in detectable consequences. Thus, for the sake of concealment and feasibility, we make the attempt to identify the critical edges whose removals will trigger the detachment of given target nodes from the $k$-core. In other words, given a set of target nodes in $k$-core, our objective is to remove a minimal-size set of edges such that all target nodes are absent from $k$-core. We refer to this problem as \underline{T}argeted $k$-\underline{N}ode\underline{s} \underline{C}ollapse \underline{P}roblem (TNsCP).

\begin{figure}[t]
  \centering
  \includegraphics[width=0.96\linewidth]{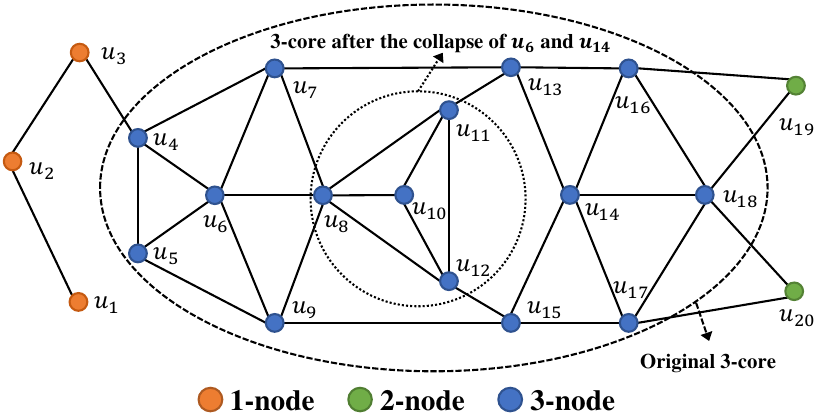}
  \caption{Example graph for introducing motivations.}
  \label{fig: motivation}
\end{figure}

\begin{example}
  We assume a social network in Figure. \ref{fig: motivation}, where nodes and edges represent the users and their relationships. Users with a high significance usually exert great attraction to others. In the $k$-core model, the significance of users is assessed by the $k$-core they belong to, e.g., users in $3$-core are usually more important than those just in $2$-core. It can be realized that the absence of $u_{6}$ and $u_{14}$ from $3$-core will greatly reduce the number of users in $3$-core, from 15 users to 4 users. With prior knowledge, we find that the breakdown of only two relationships, for example ($u_4$, $u_7$) and ($u_{16}$, $u_{18}$), will truly achieve this purpose.
\end{example}

Research on TNsCP contributes to the understanding of node vulnerability in the $k$-core model and enables maintainers to safeguard critical relationships among users. Furthermore, this study helps to diminish the importance of confidential users, making them less noticeable to potential attackers.

\textbf{Contributions.} In this paper, we provide the general definition of TNsCP with proof of its NP-hardness. Then, we introduce a novel algorithm named MOD for effective candidate reduction. Furthermore, we propose a novel heuristic algorithm named MONA to address TNsCP. We exhibit the effectiveness of MONA compared with several baselines. Finally, we show the high scalability of MONA compared with those global attack methods.


\section{Preliminaries and Problem Statement}
In this paper, we denote a graph as $G=(V,E)$, where $V$ and $E$ represent the sets of nodes and edges respectively. Note that we only focus on unweighted and undirected graphs without self-loops. We use $d_{(u, G)}$ and $N_{(u, G)}$ to represent the degree and neighbors of $u$ in $G$, respectively. In this section, we introduce some essential definitions and concepts.


\begin{definition}
  \label{def: k-core}
  \textbf{$k$-core.} Given a graph $G$ and a positive integer $k$, its $k$-core, denoted as $G_k=(V_k, E_k)$ where $V_k\subseteq V$ and $E_k\subseteq E$, means the maximal induced subgraph where each node $u \in G_k$ occupies at least $k$ neighbors.
\end{definition}

\begin{definition}
  \label{def: core-number}
  \textbf{Core Number.} Given a graph $G$ and a node $u \in G$, the core number of $u$, denoted as $C_{(u, G)}$, is the largest $k$, such that $u \in G_k$ while $u \notin G_{k+1}$.
\end{definition}

In this work, nodes with the same core number of $k$ are called \textbf{$k$-nodes}. And if the core number of a node $u$ decreases after edge removal, we refer to this event as the \textbf{collapse} of $u$. Generally, once a set of edges $\mathcal{E}$ is removed from $G$, it may lead to the collapse of multiple nodes in $G$ due to the domino phenomenon of the $k$-core \cite{goltsev2006k}. These nodes are called followers of $\mathcal{E}$ in $G$, denoted as $\mathcal{F}(\mathcal{E}, G)$. Assuming $G^\prime=G\setminus \mathcal{E}=(V^\prime, E^\prime)$, we obtain $\mathcal{F}(\mathcal{E}, G) = V_k\setminus V^\prime_k$.

\vspace{5px} 
\noindent\textbf{Problem Statement.} Given a graph $G$ and a set of target $k$-nodes $\mathcal{T} \in V_k \setminus V_{k+1}$, the \underline{T}argeted $k$-\underline{N}ode\underline{s} \underline{C}ollapse \underline{P}roblem (TNsCP) aims to remove a minimal-size set of edges $\mathcal{E} \in E$ to make all nodes in $\mathcal{T}$ collapse, i.e., $\mathcal{T} \subseteq \mathcal{F}(\mathcal{E}, G)$.


\begin{theorem}
  The TNsCP is NP-hard.
\end{theorem}
\begin{proof}
  We try to reduce the Set Covering Problem (SCP) to the TNsCP. In the SCP, two sets are given: a universe set $U$, and a set $S$ that contains subsets of $U$. The union of all subsets in $S$ could cover $U$. Then the SCP concerns finding the smallest number of subsets from $S$ to cover the entire universe \cite{GROSSMAN199781}.

  Here, for a given graph $G$ and a set of target $k$-nodes $\mathcal{T}$, we define $U = \mathcal{T}$ and $S = E$. It is important to note that $S$ itself does not represent the set of subsets of $U$. Therefore, we utilize the follower function $\mathcal{F}(\cdot, G)$ to map $S$ to $U$. For example, for an element $s \in S$, we have $\mathcal{F}(s, G) \cap U \subseteq U$. It becomes evident that $\mathcal{F}(S, G) \cap U = U$. After that, the TNsCP aims to find a minimal-size set $\mathcal{E} \subseteq S$ such that $\mathcal{F}(\mathcal{E}, G) \cap U = U$. Thus, we have successfully reduced the SCP to TNsCP. Since the NP-hardness of the SCP \cite{GROSSMAN199781}, it follows that the TNsCP is also NP-hard.
\end{proof}

\begin{theorem}
  \label{the: removal-affect1}
  Given a graph $G$ and an edge $(u, v) \in E$. Removal of $(u, v)$ could make the $k$-nodes collapse if and only if $min(d_{(u, G_k)}, d_{(v, G_k)}) = k$.
\end{theorem}

\begin{theorem}
  \label{the: removal-affect2}
  Given a graph $G$ and an edge $(u, v) \in E$. Assume that $k = C_{(u, G)} \leq C_{(v, G)}$, only those nodes with the core number of $k$ may collapse after the removal of $(u, v)$, and their core number will decrease at most 1.
\end{theorem}

The proofs of Theorem \ref{the: removal-affect1} and Theorem \ref{the: removal-affect2} are omitted here because they have already been established by \cite{10.1145/3269206.3269254} and by \cite{sariyuce2013streaming, li2013efficient}, respectively. Then, on the basis of Theorem \ref{the: removal-affect1} and Theorem \ref{the: removal-affect2}, we have the following observation.

\begin{observation}
  \label{obs: candidate}
  Those edges included in $\mathcal{P} = \{(u, v) | (u, v) \in E, min(C_{(u, G)}, C_{(v, G)}) = k\}$ are what matter in the collapse of $k$-nodes.
\end{observation}
\begin{proof}
  Given a graph $G$ and an edge $(u, v) \in E$ that satisfies $min(d_{(u, G_k)}, d_{(v, G_k)}) = k$, the followers of $(u, v)$ are the subset of $k$ nodes, that is, $\mathcal{F}(\{(u, v)\}, G) \subseteq V_{k} \setminus V_{k+1}$. After the removal of $(u, v)$, all these followers will absolutely collapse from $k$-core to $(k-1)$-core. Additionally, the collapse of $\mathcal{F}(\{(u, v)\}, G)$ will trigger the emergence of more edges that satisfy Theorem \ref{the: removal-affect1}. Following the procedure, we can iteratively remove these compliant edges until no more edges satisfy Theorem \ref{the: removal-affect1} in the graph. In such a scenario, it implies that there are no nodes with the core number of $k$ in the graph, indicating the complete collapse of the $k$-nodes. And all the edges removed in this procedure are part of $\mathcal{P} = \{(u, v) | (u, v) \in E, min(C_{(u, G)}, C_{(v, G)}) = k\}$.
\end{proof}


Based on Observation \ref{obs: candidate}, we can initially narrow the candidate edges from $E$ to $\mathcal{P}$.

\section{Methodologies}
\label{sec: methodologies}

\subsection{Optimal Solution}

Assuming that we have prior knowledge of $\mathcal{F}(\{e\}, G)$, the removal of $\{e\}$ is the optimal solution when our target nodes are within $\mathcal{F}(\{e\}, G)$. Therefore, a naive solution for TNsCP is to exhaustively enumerate all possible combinations of edges in $\mathcal{P}$ and select the best for removal. As shown in Algorithm \ref{alg: optimal}, we explore each edge combination in $\mathcal{P}$, with its size ranging from 1 to $|\mathcal{P}|$. The time complexity of Algorithm \ref{alg: optimal} is $\mathcal{O}(\sum_{i=1}^{|\mathcal{P}|}{|\mathcal{P}| \choose i}|E_k|)$, which is extremely time-consuming and unscalable with increasing input size. Due to the constraints of computational resources, it is not feasible to enumerate the followers generated by all possible edge combinations. Therefore, identifying the candidate edges that may cause the collapse of the target $k$-nodes $\mathcal{T}$ poses a challenge.

\begin{algorithm}[t]   
  \caption{Optimal($G$, $\mathcal{T}$)}         
  \label{alg: optimal}
  \SetKwInOut{Input}{input}
  \SetKwInOut{Output}{output}
  \Input{given graph $G$, target $k$-nodes $\mathcal{T}$;}               
  \Output{removed edges $\mathcal{E}$.}              

  $\mathcal{P} \leftarrow \{(u, v) | (u, v) \in E, min(C_{(u, G)}, C_{(v, G)}) = k\}$\;
  \For{$iter \leftarrow 1$ to $|\mathcal{P}|$}
  {
    $combs\leftarrow$ All combinations of size $iter$ in $\mathcal{P}$\;
    \ForEach{$\mathcal{E}\in combs$}
    {
      \textbf{if} {$\mathcal{T}\subseteq\mathcal{F}(\mathcal{E}, G)$} \textbf{then} \Return{$\mathcal{E}$\;}
    }
  }
\end{algorithm}

\subsection{Improved Candidate Reduction}

Motivated by the above challenge, in this part, we present two novel algorithms to improve candidate reduction.


Given a $k$-node $u$, the collapse of $u$ may trigger the collapse of other $k$-nodes in its neighbors. Therefore, when our target nodes are among the collapsed neighbors of $u$, it is prioritized to make $u$ collapse. This necessitates further classification of $k$-nodes to distinguish them with different collapse orders. Inspired by \underline{O}nion \underline{D}ecomposition (OD) \cite{hebert2016multi}, we propose \underline{M}odified \underline{O}nion \underline{D}ecomposition (MOD) to divide $k$-nodes of $G$ into different layers, which are illustrated in Algorithm \ref{alg: mod}.

\begin{algorithm}[tbp]   
  \caption{MOD($G$, $k$)}         
  \label{alg: mod}
  \SetKwInOut{Input}{input}
  \SetKwInOut{Output}{output}
  \Input{given graph $G$, core number constraint $k$;}               
  \Output{modified onion distribution $\mathcal{M}$.}              

  $layer \leftarrow 0$; $\tilde{G} \leftarrow G_k$; $\mathcal{S} \leftarrow \{u | u \in G, C_{(u,G)} = k\}$\;
  \While{$\mathcal{S}$ is not empty}
  {
    $layer \leftarrow layer + 1$\;
    $equal \leftarrow \{u | u \in \tilde{G}, deg(u, \tilde{G}) = k\}$\;
    $lower \leftarrow \{u | u \in \tilde{G}, deg(u, \tilde{G}) < k\}$\;
    \textbf{if} {$|lower| > 0$} \textbf{then} $equal \leftarrow \emptyset$\;
    \ForEach{$v \in lower \cup equal$}
    {
      $\mathcal{M}_v \leftarrow layer$\;
    }
    $\mathcal{S} \leftarrow \mathcal{S} \setminus \{lower \cup equal\}$\;
    $\tilde{G} \leftarrow \tilde{G} \setminus \{lower \cup equal\}$\;
  }
  \Return{$\mathcal{M}$}\;
\end{algorithm}

\begin{figure}[t]
  \centering
  \includegraphics[width=0.96\linewidth]{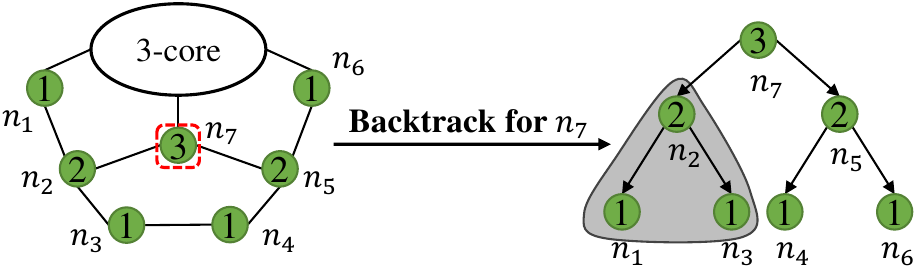}
  \caption{Example for MOD and BacktrackTree algorithms. All $2$-nodes are assigned by MOD (left) and $n_7$ is backtracked by BacktrackTree (right).}
  \label{fig: method_example}
\end{figure}

The main difference between OD and MOD is that OD removes all nodes with degrees less than or equal to $k$ in each iteration. For MOD, in Algorithm \ref{alg: mod}, we further divide these nodes into $equal$ and $lower$ (Lines 4-5). If there are nodes with degrees less than $k$, we set $equal$ to an empty set (Line 6) and then assign nodes in $lower$ with $layer$ (Lines 7-8). Otherwise, nodes in $equal$ are assigned. Then, we update $\mathcal{S}$ and $\tilde{G}$ (Lines 9-10). Finally, the modified onion distribution of $k$-nodes $\mathcal{M}$ is returned . The lower the layer of a $k$-node in $\mathcal{M}$, the higher its priority in the collapse order. For example, we execute Algorithm \ref{alg: mod} for $2$-nodes in Figure \ref{fig: method_example} (left). The number marked on each node represents its modified onion layer. In the first iteration, we remove $equal=\{n_1,n_3, n_4, n_6\}$ and assign them to $layer=1$. Then, $lower=\{n_2, n_5\}$ and $lower=\{n_7\}$ are removed in order and assigned with $layer=2$ and $layer=3$, respectively.


Given a target $k$-node $u$, we could prioritize the collapse of its neighbors with higher collapse orders according to $\mathcal{M}$. Furthermore, we can iteratively backtrack on those neighbors that maintain increasing collapse orders with respect to the current node. As shown in Algorithm \ref{alg: backtrack}, we exploit the Breadth First Search (BFS) algorithm for this backtracking procedure. In each iteration, we extract those neighbors with higher collapse orders, denoted as $neighbors$ (Lines 4-5). Then the BFS queue $queue$ is updated (line 6) and the edges that link the current node with its neighbors are added to $G_{BT}$ (line 7). Note that $G_{BT}$ is a directed graph initialized with nodes in $\mathcal{T}$ and empty edges (Line 2). For example, as shown in Figure \ref{fig: method_example} (right), we execute Algorithm \ref{alg: backtrack} with the setting of $n_7$ as the target. We obtain $G_{BT}$ with 7 nodes and 6 directed edges.


\begin{algorithm}[t]
  \caption{BacktrackTree($G$, $\mathcal{T}$)}         
  \label{alg: backtrack}
  \SetKwInOut{Input}{input}
  \SetKwInOut{Output}{output}
  \Input{given graph $G$, target $k$-nodes $\mathcal{T}$;}               
  \Output{backtrack tree $G_{BT}$.}              
  $\mathcal{M} \leftarrow$ \textbf{MOD}($G$, $k$); $queue \leftarrow \mathcal{T}$\;
  $G_{BT} \leftarrow$ a directed graph $(\mathcal{T}, \emptyset)$\;
  \ForEach{$u \in queue$}
  {
    $nodes \leftarrow \{v | v \in N(u, G_k), C_{(v, G_k)} = k\}$\;
    $neighbors \leftarrow \{v | v \in nodes, \mathcal{M}_v < \mathcal{M}_u\}$\;
    $queue \leftarrow queue \cup (neighbors \setminus V_{BT})$\;
    $G_{BT} \leftarrow G_{BT} \cup \{(u, v) | v \in neighbors\}$\;
  }
  \Return{$G_{BT}$}\;
\end{algorithm}

We can realize that those nodes in $V_{BT}$ are most probably related to the collapse of $\mathcal{T}$. And the removal of the edges in $E_{BT}$ will lead to the collapse of the nodes in $V_{BT}$. However, the edges in $E_{BT}$ are not sufficient to cover the collapse of $\mathcal{T}$ when nodes in $\mathcal{T}$ are assigned in $equal$ during Algorithm \ref{alg: mod}. Therefore, $E_{BT}$ is augmented by the adjacent edges of $\mathcal{T}$ in $G_k$. In this way, we obtain the set of candidate edges, denoted as $\mathcal{H} = E_{BT} \cup \{(u, v) | u \in \mathcal{T}, v \in N(u, G_k) \setminus V_{BT}\}$. It can be seen that $\mathcal{H} \subseteq \mathcal{P}$. Therefore, we further improve the candidates from $\mathcal{P}$ to $\mathcal{H}$.



\textbf{Complexity.} The time complexities of Algorithm \ref{alg: mod} and Algorithm \ref{alg: backtrack} are both $\mathcal{O}(|V_k \setminus V_{k+1}|)$. Additionally, both algorithms have a space complexity of $\mathcal{O}(|G_k|)$.

\subsection{MONA Algorithm}

Benefiting from the improved candidates $\mathcal{H}$, in this part, we propose a novel heuristic algorithm to tackle TNsCP, named \underline{M}odified \underline{O}nion based $k$-\underline{N}odes \underline{A}ttack (MONA).

\textbf{Pruned Followers.} Given an edge $e \in \mathcal{H}$, the removal of $e$ has two main impacts. First, as discussed before, the removal of $e$ can trigger the collapse of the nodes in $G$, that is, $\mathcal{F}(\{e\}, G)$. In addition, the removal of $e$ may result in the zero-degree of nodes in $G_{BT}$. These nodes will no longer appear in $G_{BT}$ during the backtracking process of Algorithm \ref{alg: backtrack}. This process is analogous to pruning a branch from a tree, causing all the leaves on that branch to detach from the tree. For example, as shown in Figure \ref{fig: method_example} (right), those nodes in the shadow area, i.e., $n_1$, $n_2$ and $n_3$, will absolutely be pruned after the removal of edge $(3, 2)$. These nodes will disappear in $G_{BT}$ during the re-backtrack process. We integrate these two impacts and introduce the concept of pruned followers $\mathcal{F}_p(e)$ to measure the impact of removing $e$. The operations for obtaining $\mathcal{F}_p(e)$ is presented in Algorithm \ref{alg: prune}. Note that we bypass those zero-degree nodes contained in $\mathcal{T}$.

We exhibit the details of the MONA algorithm in Algorithm \ref{alg: mona}. In Line 2 and Line 3, we first initialize $G_{BT}$ and $\mathcal{H}$. In Line 4 to Line 8, we iteratively remove the edge with the most pruned followers and then update $G_{BT}$ and $\mathcal{H}$ according to the adversarial graph $G_k \setminus \mathcal{E}$.

\begin{algorithm}[t]   
  \caption{PruneEdge($G$, $G_{BT}$, $e$, $\mathcal{T}$)}         
  \label{alg: prune}
  \SetKwInOut{Input}{input}
  \SetKwInOut{Output}{output}
  \Input{given graph $G$, backtrack tree $G_{BT}$, removed edge $e$, target $k$-nodes $\mathcal{T}$;}               
  \Output{pruned followers $\mathcal{F}_p(e)$.}              
  $\tilde{G} \leftarrow G_{BT} \setminus \{e\}$; $\mathcal{F}_p \leftarrow \mathcal{F}(\{e\}, G) \cap V_{BT}$\;
  \While{exist zero-indegree nodes in $\tilde{V} \setminus \mathcal{T}$}
  {
    $\mathcal{Z} \leftarrow $ All zero-indegree nodes in $\tilde{V} \setminus \mathcal{T}$\;
    $\tilde{G} \leftarrow \tilde{G} \setminus \mathcal{Z}$; $\mathcal{F}_p \leftarrow \mathcal{F}_p  \cup \mathcal{Z}$\;
  }
  \Return{$\mathcal{F}_p(e)$}\;
\end{algorithm}

\begin{algorithm}[t]   
  \caption{MONA($G$, $\mathcal{T}$)}         
  \label{alg: mona}
  \SetKwInOut{Input}{input}
  \SetKwInOut{Output}{output}
  \Input{given graph $G$, target $k$-nodes $\mathcal{T}$;}               
  \Output{removed edges $\mathcal{E}$.}              

  $\mathcal{E} \leftarrow \emptyset$; $G_{BT} \leftarrow$ \textbf{BacktrackTree}($G_k $, $\mathcal{T}$)\;
  $\mathcal{H} \leftarrow E_{BT} \cup \{(u, v) | u \in \mathcal{T}, v \in N(u, G_k) \setminus V_{BT}\}$\;

  \While{$\mathcal{T} \not\subseteq \mathcal{F}(\mathcal{E}, G)$}
  {
    \ForEach{$e \in \mathcal{H}$}
    {
      $\mathcal{F}_p(e) \leftarrow$ \textbf{PruneEdge}($G_k \setminus \mathcal{E}$, $G_{BT}$, $e$, $\mathcal{T}$)\;
    }

    $e^{\star} \leftarrow$ The edge with most pruned followers\;
    $\mathcal{E} \leftarrow \mathcal{E} \cup \{e^\star\}$; Update $G_{BT}$ and $\mathcal{H}$ with $G_k \setminus \mathcal{E}$\;

  }

  \Return{$\mathcal{E}$}\;
\end{algorithm}

\textbf{Complexity.} The time complexities of Algorithm \ref{alg: prune} and Algorithm \ref{alg: mona} are $\mathcal{O}(|V_{BT}| + |E_k|)$ and $\mathcal{O}(|V_k \setminus V_{k+1}| + |\mathcal{H}| \cdot |V_{BT}| + |\mathcal{H}| \cdot |E_k|)$, respectively. Additionally, both algorithms have a space complexity of $\mathcal{O}(|G_k|)$.

\section{Experiments}

\subsection{Experimental Settings}
\textbf{Datasets.} Basic properties of used datasets are presented in Table \ref{tab: dataset}. All these datasets are collected from \cite{nr}. Note that all networks are converted to undirected and unweighted graphs without self-loops.

\begin{table}[htbp]
  \caption{Basic properties of used datasets, where $d_{avg}$ is the average degree and $k_{max}$ is the maximal core number.}
  \label{tab: dataset}
  \setlength{\tabcolsep}{0.5mm}
  \renewcommand{\arraystretch}{1.2}
  \begin{tabular*}{\hsize}{@{\extracolsep{\fill}}l|rrrrrr}
    \bottomrule[0.5mm]
    Dataset     & $|V|$     & $|E|$      & $d_{avg}$ & $k_{max}$ & $|V_{k_{max}}|$ & $|E_{k_{max}}|$ \\ \hline
    USAir       & 332       & 2,126      & 12.81     & 26        & 35          & 539         \\
    DeezerEU    & 28,281    & 92,752     & 6.56      & 12        & 71          & 564         \\
    Crawl       & 1,112,702 & 2,278,852  & 4.10      & 18        & 725         & 11,522      \\
    YouTube     & 1,134,890 & 2,987,624  & 5.27      & 51        & 845         & 36,363      \\
    Lastfm      & 1,191,805 & 4,519,330  & 7.58      & 70        & 597         & 35,153      \\
    Wikipedia   & 1,864,433 & 4,507,315  & 4.84      & 66        & 324         & 16,054      \\
    Roadnet     & 1,957,027 & 2,760,388  & 2.82      & 3         & 4,454       & 7,393       \\
    Talk        & 2,394,385 & 4,659,565  & 3.89      & 131       & 700         & 73,503      \\
    Patent      & 3,774,768 & 16,518,947 & 8.75      & 64        & 106         & 4,043       \\
    Livejournal & 4,033,137 & 27,933,062 & 13.85     & 213       & 214         & 22,791      \\
    \toprule[0.5mm]
  \end{tabular*}
\end{table}


\textbf{Baselines.} The following baseline methods are considered for evaluations and comparisons with our proposed method.

\begin{itemize}
  \item \textbf{Random} randomly removes an edge from $\mathcal{P}$ in each iteration until all nodes in $\mathcal{T}$ collapse.
  \item \textbf{Degree} removes the edge with the lowest degree (the sum of the degrees of its endpoints in $G_k$) in $\mathcal{P}$ in each iteration until all nodes in $\mathcal{T}$ collapse.
  \item \textbf{COREATTACK} \cite{zhou2022attacking} iteratively removes the edge with most followers in $G_{k_{max}}$ until all $k_{max}$-nodes collapse.
  \item \textbf{KC-Edge} \cite{10.1145/3269206.3269254} iteratively removes the edge with most followers in $G_k$ until its perturbation budget $p$ is fulfilled or $G_k$ is empty.
\end{itemize}

All programs are implemented in Python 3.7.11. All experiments are carried out on a machine equipped with Intel(R) Xeon(R) Gold 5218R CPU @2.10GHz and Linux Ubuntu 20.04.4.  Note that we select the \textbf{top-$b$} nodes with the highest degree within $G_{k_{max}}$ as our target nodes $\mathcal{T}$. And Random is performed 100 times independently for each task and the mean value is recorded.

\subsection{Effectiveness of Improved Candidate Reduction}

To evaluate the effectiveness of the improved candidate reduction proposed in Section \ref{sec: methodologies}, we first compare the number of candidate edges utilized by Optimal and MONA, i.e. $|\mathcal{P}|$ and $|\mathcal{H}|$. We implement these two algorithms in USAir and DeezerEU with $b$ varying from 2 to 10. As shown in Figure \ref{fig: candidate_a} and \ref{fig: candidate_b}, a notable distinction is observed in their candidate sizes. Moreover, to further justify the approximation guarantee of MONA, we also compare its attack performance, i.e., the number of removed edges $|\mathcal{E}|$, with that of Optimal. The results in Figures \ref{fig: delete_a} and \ref{fig: delete_b} support that MONA could achieve the optimal solution. Furthermore, we also illustrate the candidate reduction for the other datasets with the setting of $b=30$, which is shown in Figure \ref{fig: delete_all}.


\begin{figure}[h]
  \centering
  \includegraphics[width=0.4\linewidth]{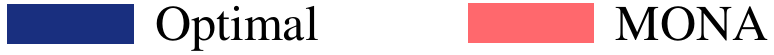}\\
  \subfigure[Candidate edges on USAir]{
    \includegraphics[width=0.46\linewidth]{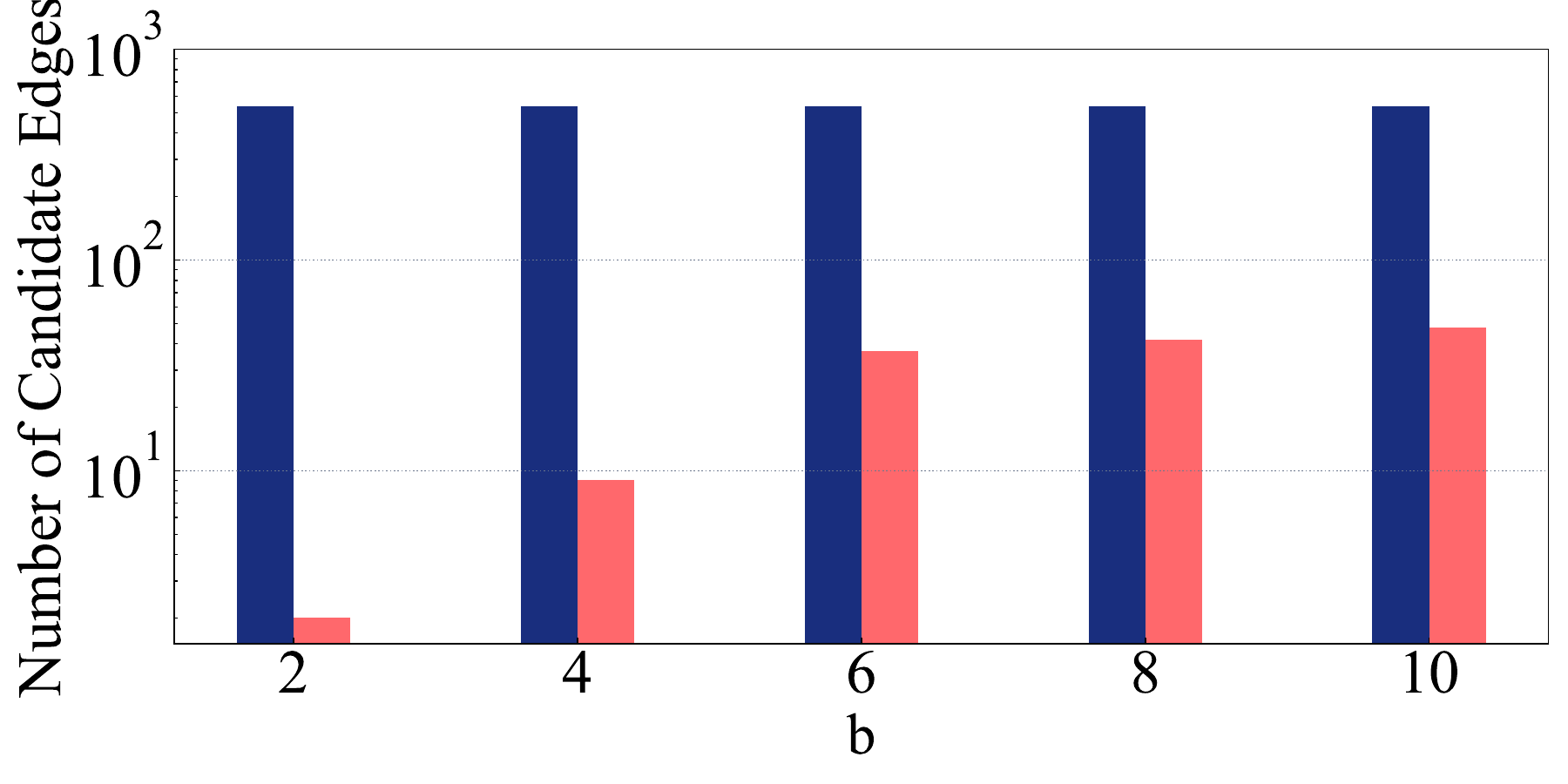}\label{fig: candidate_a}
  }
  \subfigure[Candidate edges on DeezerEU]{
    \includegraphics[width=0.46\linewidth]{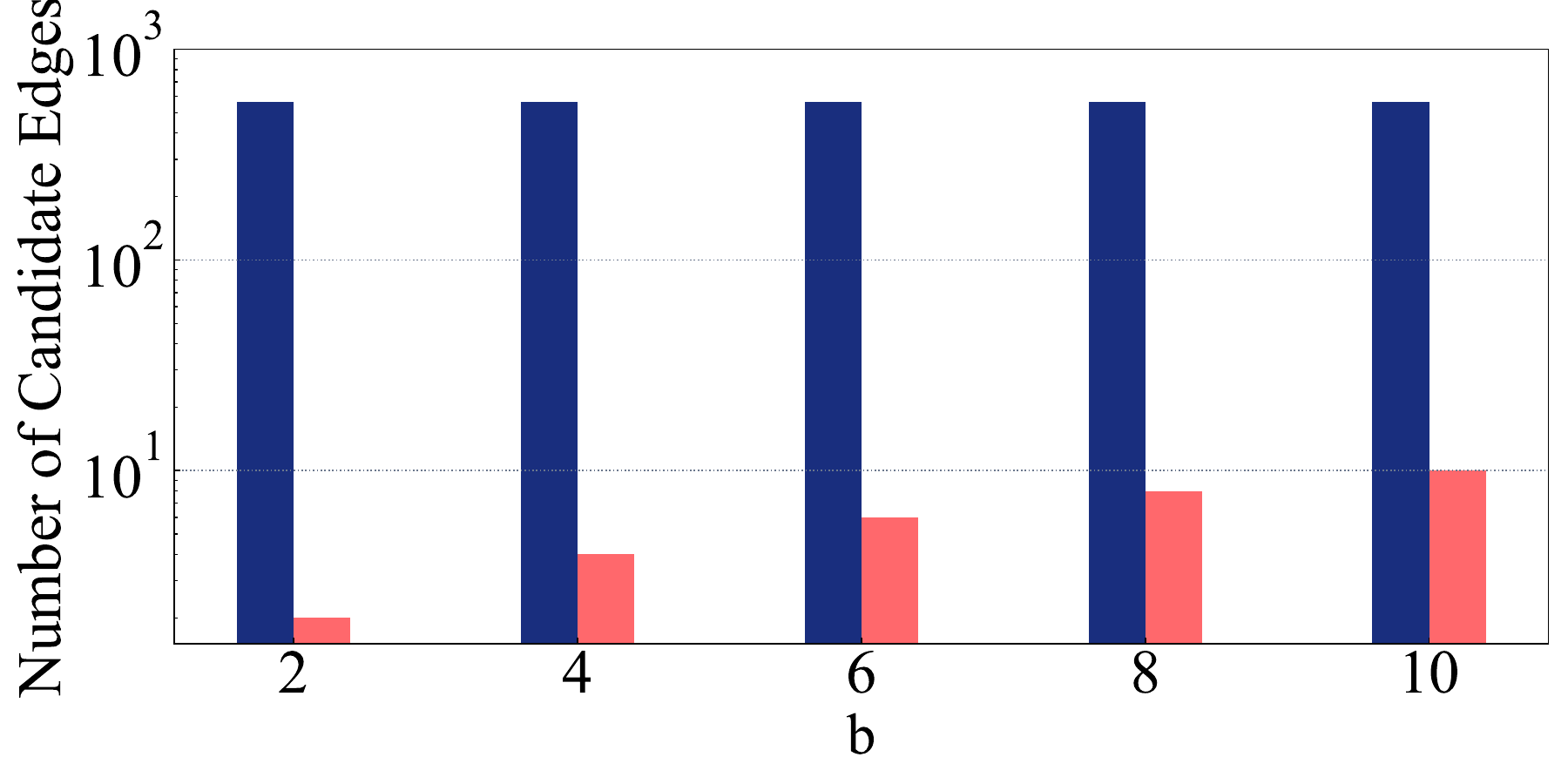}\label{fig: candidate_b}
  }
  \subfigure[Removed edges on USAir]{
    \includegraphics[width=0.46\linewidth]{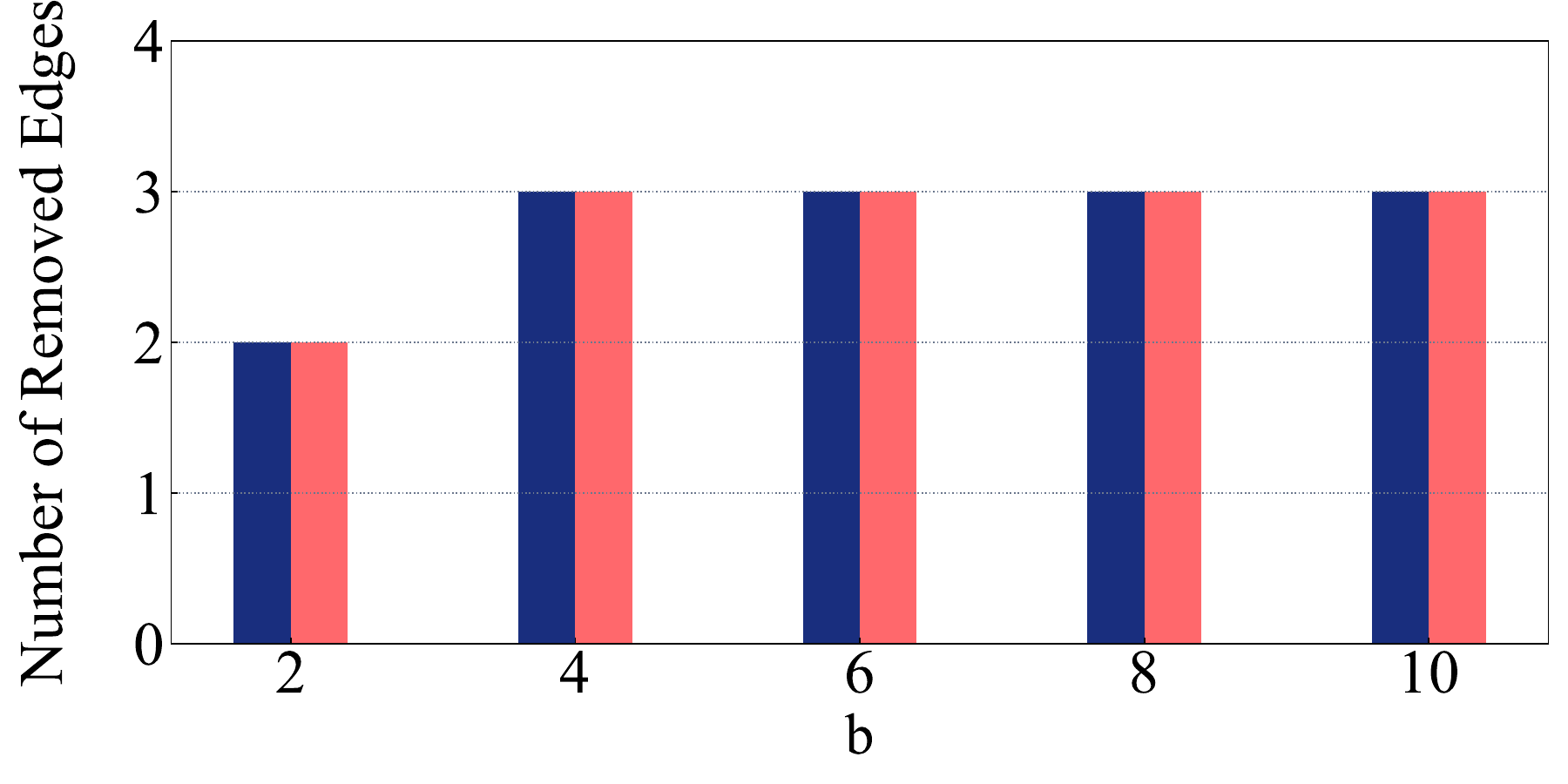}\label{fig: delete_a}
  }
  \subfigure[Removed edges on DeezerEU]{
    \includegraphics[width=0.46\linewidth]{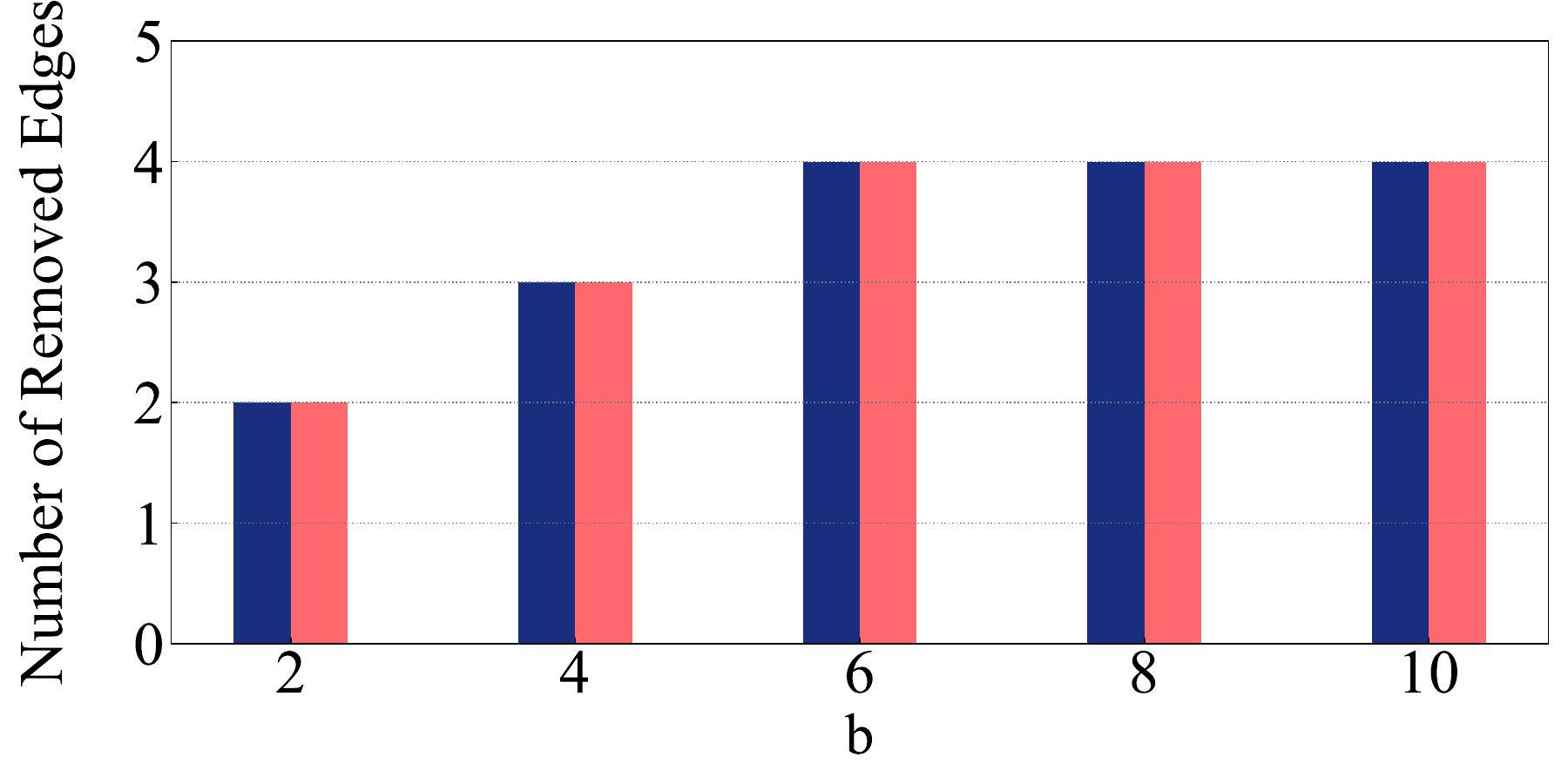}\label{fig: delete_b}
  }
  \subfigure[Candidate edges on the rest datasets with $b=30$.]{
    \includegraphics[width=0.96\linewidth]{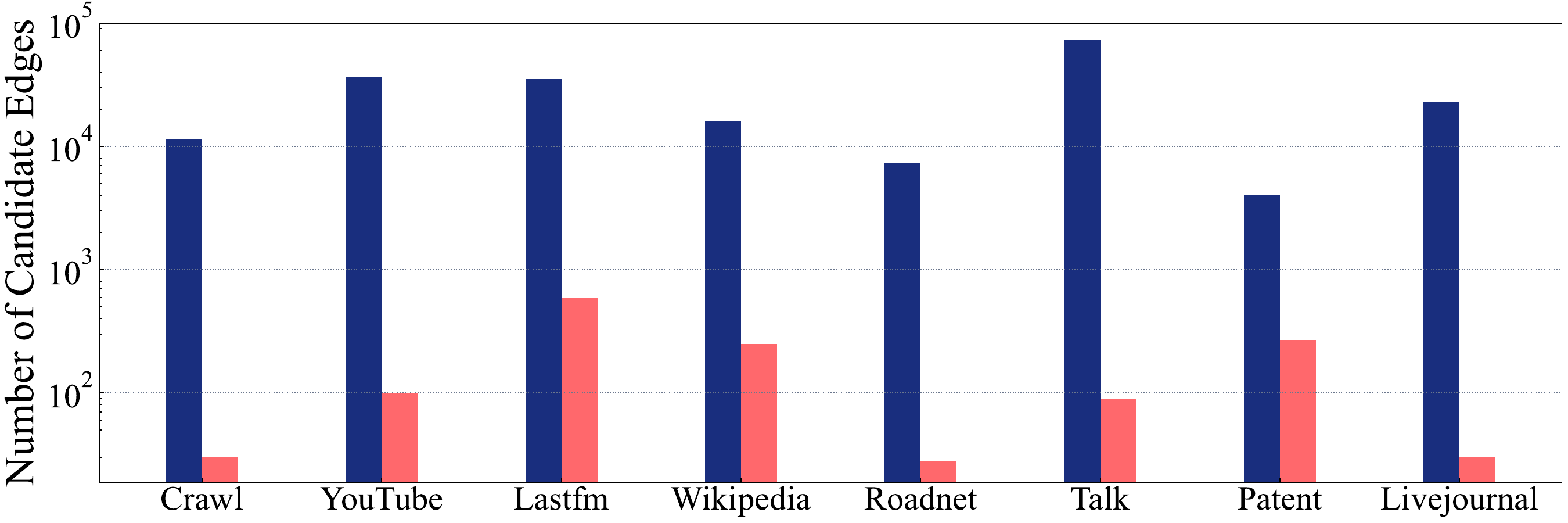}\label{fig: delete_all}
  }
  \caption{Effectiveness of improved candidate reduction with different settings of hyperparameter $b$.}
  \label{fig: comparisons}
\end{figure}


\subsection{Effectiveness of MONA Algorithm}

Furthermore, we evaluate the effectiveness of MONA algorithm with the comparisons of Random and Degree  in the setting of $b=30$ in Figure \ref{fig: effectiveness}. Specifically, the results in Figure \ref{fig: large_delete} illustrate that MONA consistently removes the fewest edges for the collapse of target nodes across all tested datasets. It is worth noting that on Livejournal, all three methods end up with only one edge removed. This is due to the fact that the $G_{k_{max}}$ of Livejournal is a complete graph, where the removal of any edge results in the collapse of all $k_{max}$-nodes. Since Random and Degree do not traverse all candidate edges to find the best one in each iteration, it is not surprising that they have a lower time complexity than MONA in Figure \ref{fig: time_cost}. However, MONA achieves more precise results than Random and Degree in acceptable time cost. Additionally, Figure \ref{fig: trend} presents the number of removed edges under different settings of $b$ on Lastfm and Wikipedia. The number of removed edges clearly increases with increasing $b$ for MONA. However, the upward trend becomes slower when $b$ reaches 30.


\begin{figure}[h]
  \centering
  \includegraphics[width=0.55\linewidth]{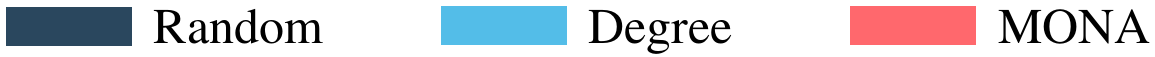}\\
  \subfigure[Comparisons of removed edges on 8 datasets]{
    \includegraphics[width=0.96\linewidth]{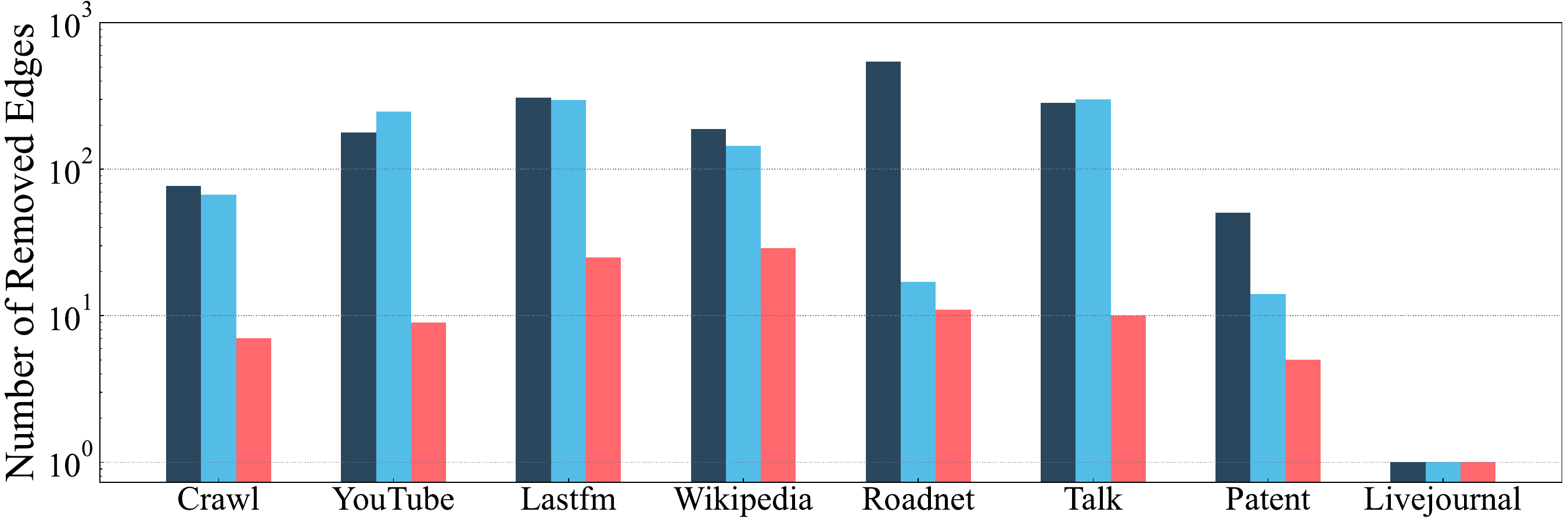}
    \label{fig: large_delete}
  }
  \subfigure[Comparisons of time cost on 8 datasets]{
    \includegraphics[width=0.96\linewidth]{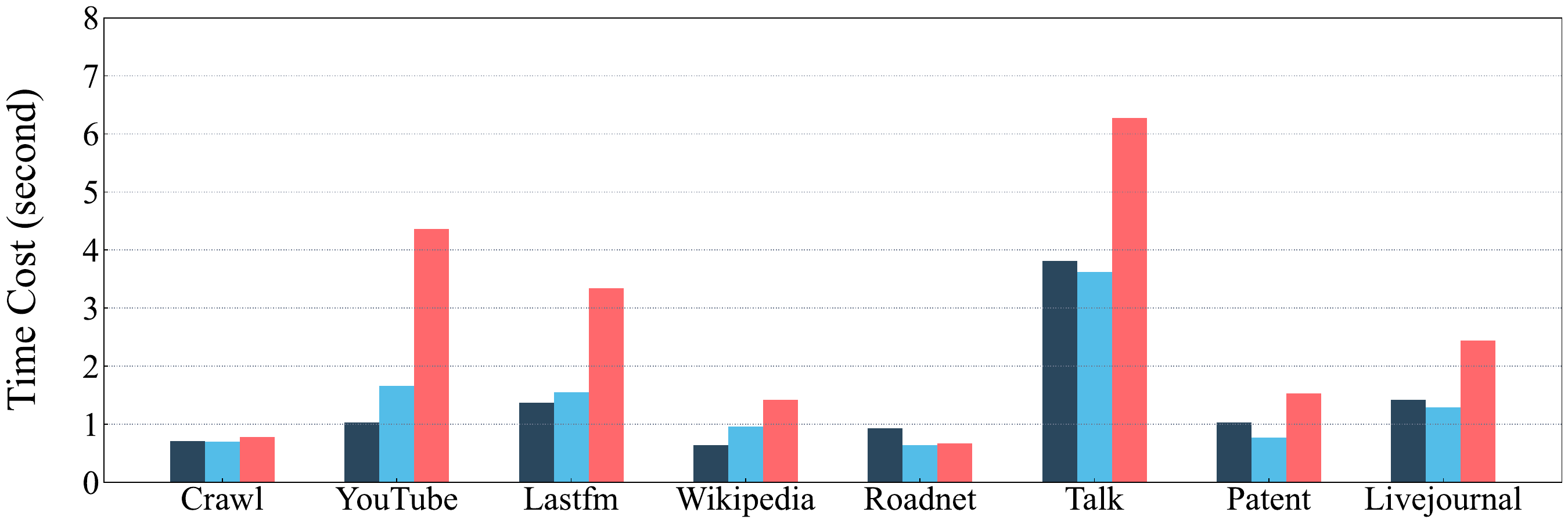}
    \label{fig: time_cost}
  }
  \caption{Effectiveness and efficiency of MONA compared with Random and Degree with the setting of $b=30$.}
  \label{fig: effectiveness}
  \vspace{-4mm}
\end{figure}

\begin{figure}[h]
  \centering
  \includegraphics[width=0.55\linewidth]{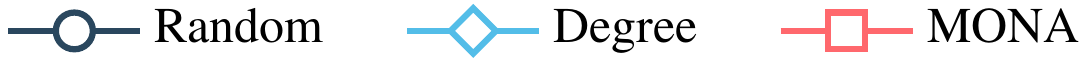}\\
  \subfigure[Lastfm]{
    \includegraphics[width=0.46\linewidth]{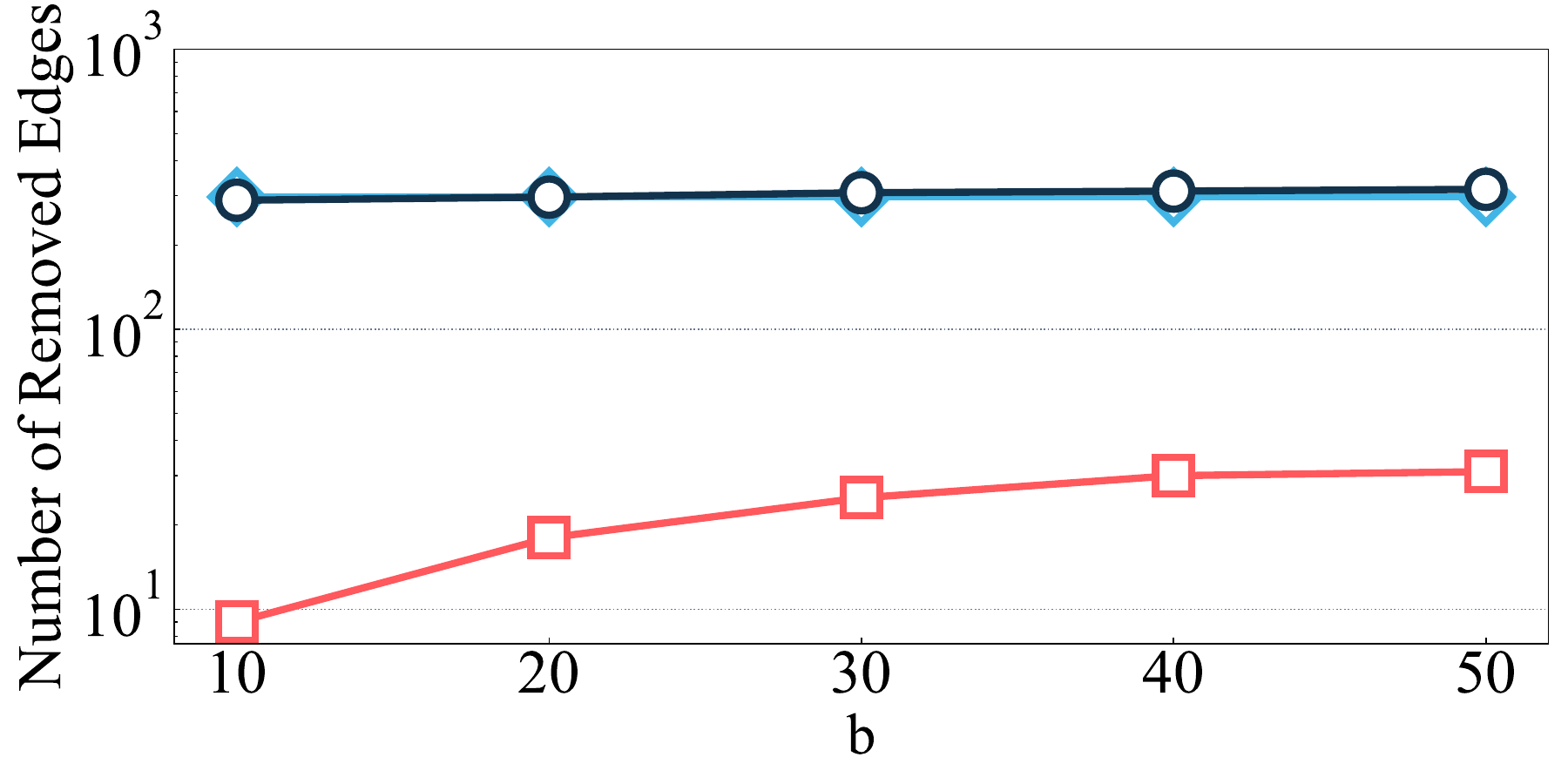}
  }
  \subfigure[Wikipedia]{
    \includegraphics[width=0.46\linewidth]{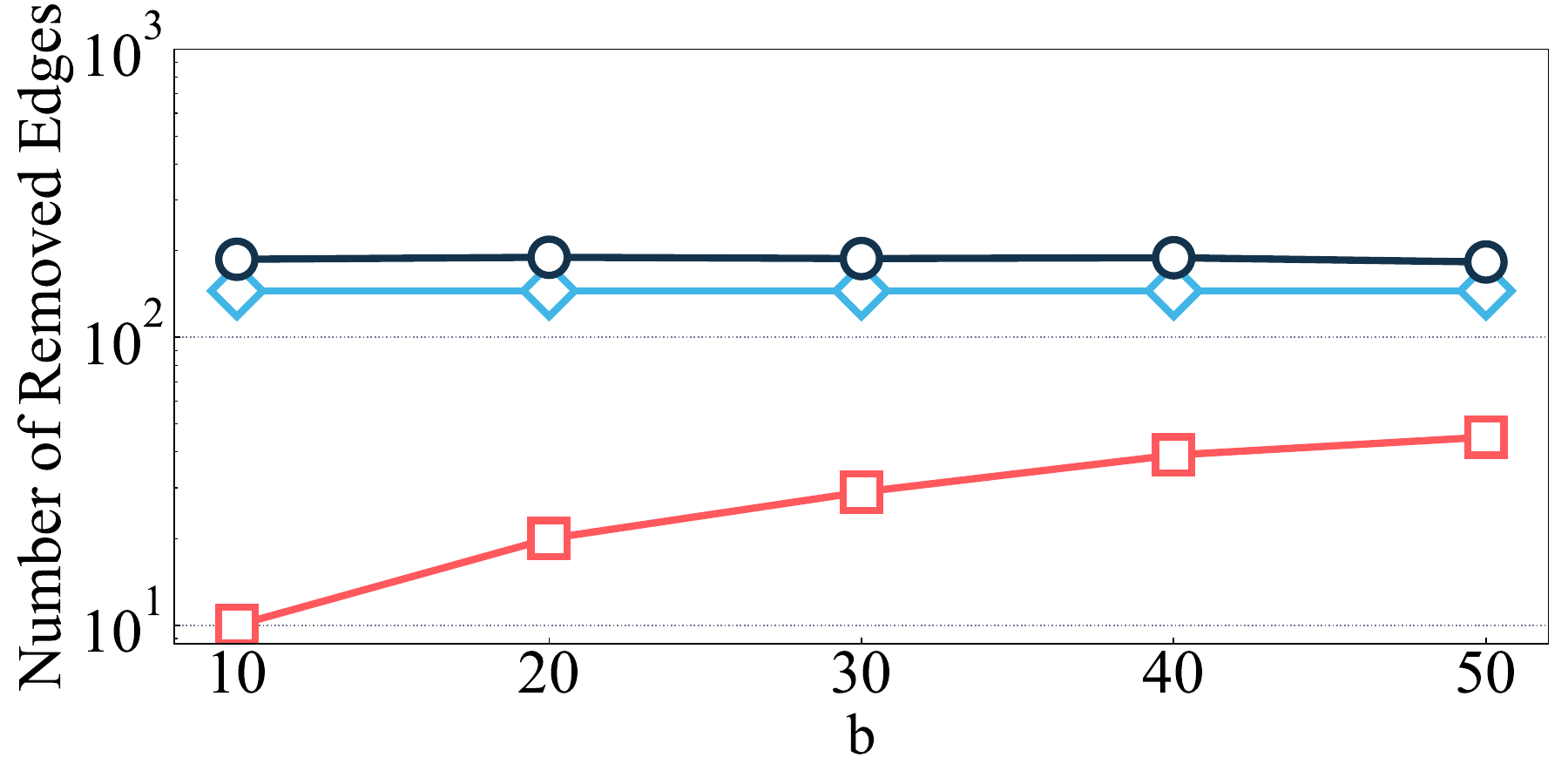}
  }
  \caption{Removed edges on Lastfm and Wikipedia with $b$ from 10 to 50.}
  \label{fig: trend}
  \vspace{-4mm}
\end{figure}

\subsection{Scalability of MONA for Global Attack}

In previous works, Zhou et al. \cite{zhou2022attacking} and Zhang et al. \cite{10.1145/3269206.3269254} respectively designed global-perspective $k$-core attack methods, namely COREATTACK and KC-Edge. By increasing the number of target nodes, we can easily extend MONA to the realm of global attacks. To compare with COREATTACK, we set $b=|E_{k_{max}}|$. And to compare with KC-Edge, we set $b=|E_{k_{max}}|$ and terminate MONA when the number of removed edges reaches $|\mathcal{E}|=p=10$. The comparative results are shown in Figure \ref{fig: scability}. In global settings, it is evident that the MONA algorithm exhibits consistent effectiveness with these global attack methods. It demonstrates the scalability of MONA for global attack.


\begin{figure}[t]
  \centering
  \includegraphics[width=0.6\linewidth]{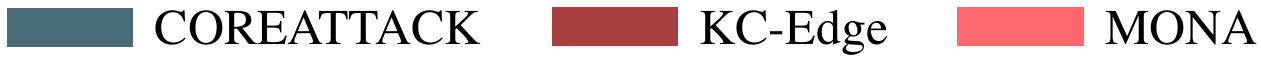}\\
  \subfigure[MONA \textit{vs} COREATTACK]{
    \includegraphics[width=0.46\linewidth]{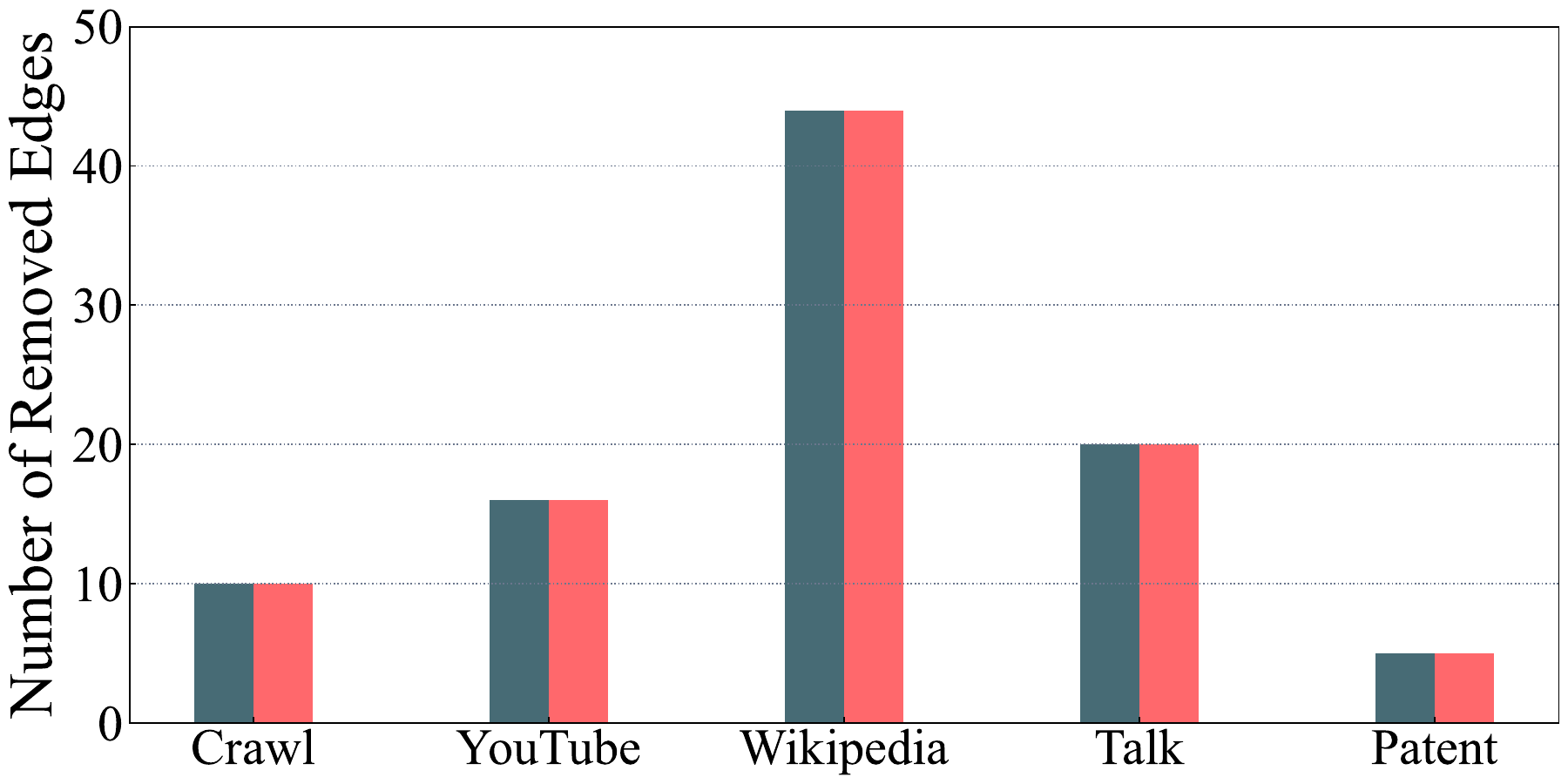}
  }
  \subfigure[MONA \textit{vs} KC-Edge, $p=10$]{
    \includegraphics[width=0.46\linewidth]{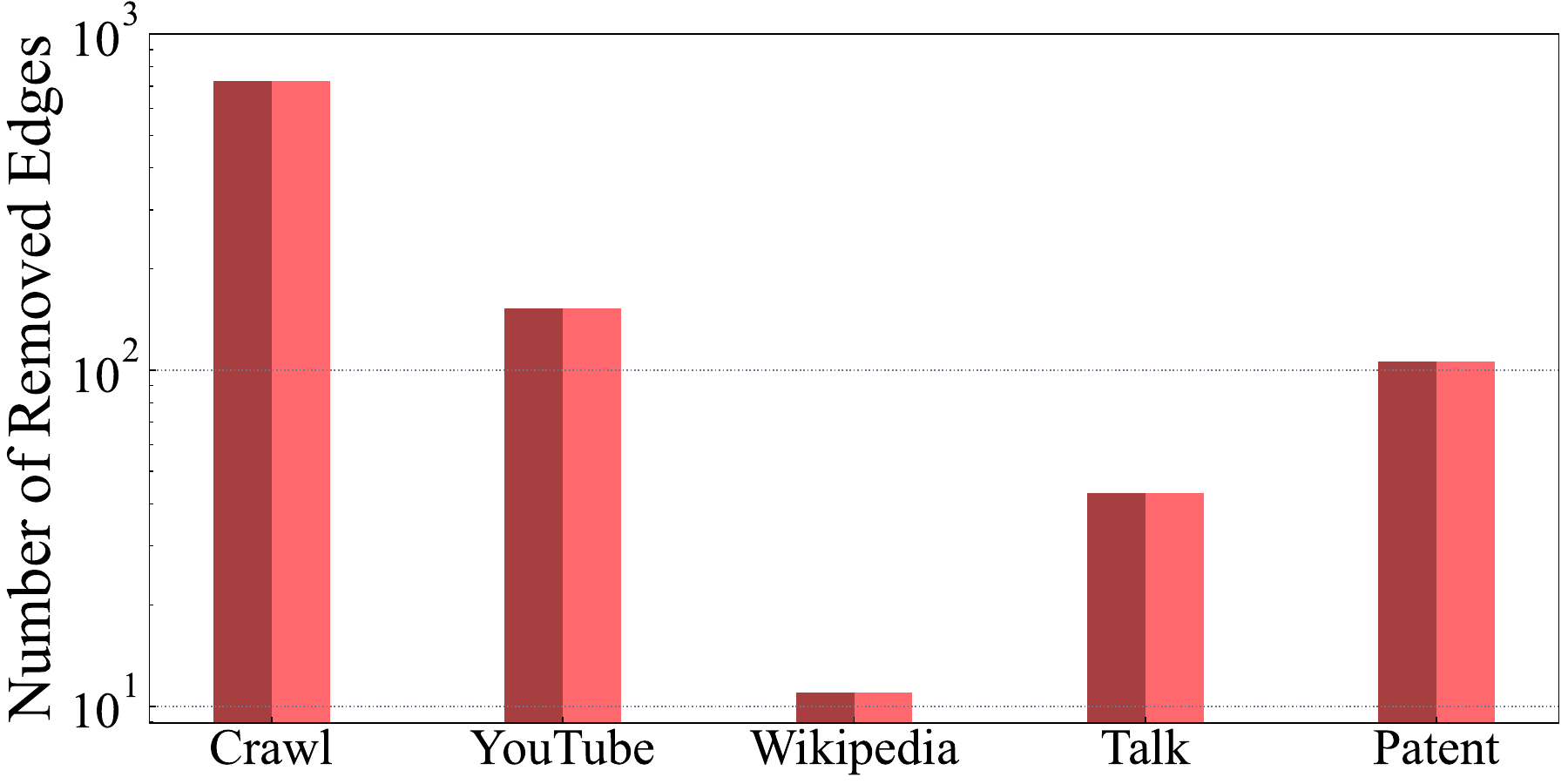}
  }
  \caption{The scalability of MONA for global attack  with the comparisons of COREATTACK and KC-Edge.}
  \label{fig: scability}
\end{figure}

\section{Conclusion}

In this paper, we propose and study TNsCP, which aims to remove a minimal-size set of edges for the collapse of target $k$-nodes. And we offer a proof of its NP-hardness. To improve time complexity, we provide a novel algorithm named MOD for candidate reduction. Furthermore, on the basis of MOD, an efficient heuristic algorithm named MONA is proposed to address TNsCP. Extensive experiments on 10 real-world datasets demonstrate the effectiveness and scalability of MONA compared to multiple baselines. 


\bibliographystyle{IEEEtran}
\bibliography{references}

\end{document}